\documentclass[letterpaper,UKenglish,cleveref, autoref]{lipics-v2019}

\usepackage{mathrsfs}
\usepackage{mathtools}
\usepackage{amsthm}
\usepackage{amssymb}
\usepackage{changepage}
\usepackage{tikz}
\usepackage[inline,shortlabels]{enumitem}
\usepackage{ stmaryrd }
\usepackage{thmtools}
\usepackage{thm-restate}
\usepackage{tipa}
\usepackage{todonotes}
\usepackage{microtype}
\usepackage{chngcntr}
\usepackage{xxcolor}

\usetikzlibrary{chains,matrix,scopes,decorations.shapes,arrows,shapes}

\usepackage{wrapfig}
\usepackage{scrextend}

\counterwithin*{equation}{section}
\counterwithin*{equation}{subsection}

\makeatletter
\let\epsilon\varepsilon
\let\phi\varphi

\makeatother

\usetikzlibrary{positioning,fit,arrows,automata,calc,shapes,decorations.pathmorphing} 

\newtheorem{innercustomclaim}{Claim}
\newenvironment{customclaim}[1]
  {\renewcommand\theinnercustomclaim{#1}\innercustomclaim}
  {\endinnercustomclaim}

\theoremstyle{definition}

\title{A Ramsey Theorem for Finite Monoids} 

\author{Isma\"el Jecker}{Institute of Science and Technology, Austria}{ismael.jecker@gmail.com}{}{}{}

\authorrunning{I. Jecker}

\Copyright{Isma\"{e}l Jecker}



\category{}

\relatedversion{}

\supplement{}

\funding{This project has received funding from the European Union’s Horizon 2020 research and innovation programme under the Marie Skłodowska-Curie Grant Agreement No. 754411.}

\acknowledgements{I wish to thank
Michaël Cadilhac,
Emmanuel Filiot and
Charles Paperman for their valuable insights
concerning Green's relations.}

\nolinenumbers 

\hideLIPIcs  

\newcommand{\vari}{k}

\newcommand{\neutral}[1]{1_{#1}}

\newcommand{\maxSemigroup}[1]{H_{#1}}

\newcommand{\firstintegers}[1]{[1, #1 ]}

\newcommand{\transformation}[1]{T_{#1}}
\newcommand{\relation}[1]{B_{#1}}

\newcommand{\product}[1]{\pi(#1)}

\newcommand{\refrel}{\shortrightarrow}

\newcommand{\truevalR}[1]{\textsf{R}_{#1}}
\newcommand{\truevalRam}[2]{\textsf{R}_{#1}(#2)}

\newcommand{\reglength}[1]{L(#1)}

\newcommand{\idempower}[1]{{#1}^{\#}}

\newcommand{\calM}{M}

\newcommand{\calJ}{\mathcal{J}}
\newcommand{\calH}{\mathcal{H}}

\newcommand{\calD}{\mathcal{D}}

\newcommand{\calG}{\mathcal{G}}
\newcommand{\calS}{S}

\begin{document}

\maketitle

\begin{abstract}
Repeated idempotent elements are commonly used
to characterise iterable
behaviours in abstract models of computation.
Therefore, given a monoid $\calM$,
it is natural to ask how long a sequence of elements
of $\calM$ needs to be to ensure the presence
of consecutive idempotent factors.
This question is formalised through the notion
of the \emph{Ramsey function} $\truevalR{\calM}$
of a finite monoid $\calM$,
obtained by mapping every $\vari \in \mathbb{N}$
to the minimal integer $\truevalRam{\calM}{\vari}$
such that every word $u \in \calM^*$
of length $\truevalRam{\calM}{\vari}$ contains
$\vari$ consecutive non-empty factors
that correspond to the same idempotent element of $\calM$.

In this work,
we study the behaviour of the Ramsey function
$\truevalR{\calM}$
by investigating the
\emph{regular $\calD$-length} of $\calM$,
defined as the largest size $\reglength{\calM}$ of a submonoid of $\calM$
isomorphic to
the set of natural numbers $\{1,2, \ldots, \reglength{\calM}\}$
equipped with the $\max$ operation.
We show that the regular $\calD$-length of $\calM$
determines the degree of
$\truevalR{\calM}$,
by proving that
$\vari^{\reglength{\calM}} \leq
\truevalRam{\calM}{\vari} \leq
(\vari|\calM|^{4})^{\reglength{\calM}}$.

To allow applications of this result,
we provide the value of the regular $\calD$-length of diverse monoids. 
In particular,
we prove that
the full monoid of $n \times n$ Boolean matrices,
which is used to express transition monoids of non-deterministic automata,
has a regular $\calD$-length
of $\frac{n^2+n+2}{2}$.
\end{abstract}

\section{Introduction}
The algebraic approach to language theory
was initiated by
Sch\"utzenberger
with the definition of the
syntactic monoid associated to a
formal language~\cite{schutzenberger1955theorie}.
This led to several parallels being drawn between
classes of languages and varieties of monoids,
the most famous being that
rational languages are characterised by finite
syntactic monoids~\cite{RabinS59},
and that star-free languages are characterised by
finite aperiodic syntactic monoids~\cite{schutzenberger1965finite}.
These characterisations motivate the study of finite monoids
as a way to gain some insight about automata.
In this work, we focus on the following problem:
\smallskip
\noindent
\begin{adjustwidth}{3pt}{3pt}
\emph{Given a finite monoid $\calM$ and $\vari \in \mathbb{N}$,
what is the minimal integer $\truevalRam{\calM}{\vari}$
such that every word $u \in \calM^*$ of length $\truevalRam{\calM}{\vari}$
contains $\vari$ consecutive factors
corresponding to the same idempotent element of $\calM$?}
\end{adjustwidth}
\smallskip
The interest of this problem lies in the fact that
when we model the behaviours of an abstract machine as elements of a monoid,
repeated idempotent factors often characterise
the behaviours that
have good properties with respect to iteration.
This can be used, for instance, to obtain pumping lemmas,
as seen in~\cite{MazowieckiR18} for weighted automata.

A partial answer to this problem is obtained by using
Ramsey's Theorem~\cite{ramsey2009problem}
or Simon's Factorisation Forest Theorem~\cite{Simon90}
(these techniques are detailed in Appendix \ref{subsec:comparison}),
as both approaches provide upper bounds for
$\truevalRam{\calM}{\vari}$.
However, neither approximation is precise:
Ramsey's theorem disregards the monoid structure,
and the Factorisation Forest Theorem
guarantees much more than
what is required here.
We prove a version of Ramsey's Theorem
adapted to monoids,
or, equivalently,
a weaker version of the Forest Factorisation Theorem,
that yields an improved bound relying
on a parameter of monoids called the regular $\calD$-length.
We now present some examples,
followed with an overview of the main concepts studied in this paper:
the Ramsey function associated to a monoid
and the regular $\calD$-length.

\subsection{Examples}
We describe three families of monoids,
along with the corresponding idempotent elements.

\subparagraph*{Max monoid}
    The max monoid $\maxSemigroup{n}$ is the set
    $\{1,2,\ldots,n\}$, equipped with the max operation.
    In this monoid, every element $i$ is idempotent
    since $\max(i,i) = i$.
\subparagraph*{Transformation monoid}
    The (full) transformation monoid $\transformation{n}$
    is the set of all (partial) functions from a set of $n$ elements into itself,
    equipped with the composition.
    See~\cite{ganyushkin2008classical}
    for a detailed definition
    of $\transformation{n}$ and its properties.
    Transformation monoids contain a wide range of idempotent elements.
    For instance, the identity function,
    mapping each element
    to itself, or the constant function $f_i$,
    mapping all elements to one fixed element $i$,
    are idempotent.
    In general,
    a function $f$ is idempotent if and only if each element $i$
    of its range satisfies $f(i)=i$.
    Transformations are commonly used to express
    transition monoids of deterministic finite state automata,
    as in this setting
    each input letter acts as a function over the set of states.
\subparagraph*{Relation monoid}
    For non-deterministic automata,
    transition monoids are more complex:
    functions fail to model
    the behaviour of the input letters
    since a single state can transition
    towards several distinct states.
    We use the (full) relation monoid $\relation{n}$ of
    all $n \times n$ Boolean matrices
    (matrices with values in $\{0,1\}$),
    equipped with the usual matrix composition
    (considering that $1+1 = 1$).
    There are plenty of idempotent matrices,
    for instance every diagonal matrix,
    or the full upper triangular matrix.
    Idempotent Boolean matrices are characterised in~\cite{Mukherjea1980},
    they correspond to specific orders
    over the subsets of $\{1,2,\ldots,n\}$.

\subsection{Ramsey function}
Given a finite monoid $\calM$,
the \emph{Ramsey function} $\truevalR{\calM}$ associated to $\calM$
maps each $\vari \in \mathbb{N}$
to the minimal integer
$\truevalRam{\calM}{\vari}$
such that every sequence of elements of $\calM$
of length $\truevalRam{\calM}{\vari}$
contains $\vari$ non-empty consecutive factors
that all correspond to the same idempotent element of $\calM$.

\subparagraph{Related work}
There are several known methods to approximate 
the Ramsey function $\truevalR{\calM}$
of a monoid $\calM$.
Ramsey's Theorem
and Simon's Factorisation Forest Theorem
are commonly used,
however, as stated before,
these approaches are too general
to obtain a precise bound.
The value of $\truevalRam{\calM}{\vari}$
is studied in~\cite{hall1996idempotents}
in the particular case $\vari=1$.
The authors prove that for a monoid $\calM$
that contains $N$ non-idempotent elements,
$\truevalRam{\calM}{1} \leq 2^N-1$.
No general related lower bound is proved,
but they show that for every $N \in \mathbb{N}$,
there exists a monoid $\calM_N$
with $N$ non-idempotent elements
that actually reaches the upper bound:
$\truevalRam{\calM_N}{1} = 2^N-1$.

\subparagraph{Our contributions}
We prove new bounds for
$\truevalR{\calM}$
by following a different approach:
instead of focusing on the non-idempotent elements of $\calM$,
we study its idempotent elements,
and the way in which they interact.
In Section \ref{sec:ramDec},
we start by considering two specific cases
where the exact value of the Ramsey function is easily obtained.
First, for a group $\calG$,
the Ramsey function is polynomial with respect to the size of $\calG$:
$\truevalRam{\calG}{\vari} = \vari|\calG|$.
Second,
we call \emph{max monoid} $\maxSemigroup{n}$
the set $\{1,2,\ldots,n\}$ equipped with the max operation,
and we show that here
the Ramsey function is exponential with respect to the size of $\maxSemigroup{n}$:
$\truevalRam{\maxSemigroup{n}}{\vari} = \vari^{n}$.
The later result implies that
$\vari^{n}$ is a lower bound for the Ramsey
function of every monoid $\calM$
that has $H_n$
as a submonoid.
Motivated by this observation,
we show how to get a related upper bound:
We define the \emph{regular $\calD$-length} $\reglength{\calM}$
of $\calM$ as the size of the largest max monoid $\maxSemigroup{\reglength{\calM}}$
embedded in $\calM$, and prove the following result.
\begin{restatable}{theorem}{ThmValRam}
\label{theorem:general_valram}
Every monoid $\calM$
of regular $\calD$-length $L$ satisfies
$
\vari^{L} \leq
\truevalRam{\calM}{\vari} \leq
(\vari|\calM|^{4})^{L}
$.
\end{restatable}
\noindent
Stated differently: every word $u \in \calM^*$
of length $(\vari|\calM|^{4})^{L}$
contains $\vari$ consecutive non-empty factors
corresponding to the same idempotent element of $\calM$,
and, conversely, there exists a word $u_\calM \in \calM^*$
of length $\vari^{L}-1$
that does not contain $\vari$ consecutive non-empty factors
corresponding to the same idempotent element.
Note that while the gap between the lower and upper bound is still wide,
this shows that the degree of
the Ramsey function $\truevalR{\calM}$
is determined by
the regular $\calD$-length of $\calM$.

\subsection{Regular $\calD$-length}\label{sec:regLength}
Theorem \ref{theorem:general_valram} states that
the degree of the Ramsey function of a monoid $\calM$
is determined by the regular $\calD$-length of $\calM$,
which is
the size of the largest max monoid
embedded in $\calM$.
We now show that for transformation monoids and
relation monoids,
the regular $\calD$-length is exponentially shorter than
the size.
Let us begin by mentioning an equivalent definition
of the regular $\calD$-length
in terms of Green's relations.
While this alternative definition is not used
in the proofs presented in this paper,
it allows us to immediately obtain the regular $\calD$-length
of monoids whose Green's relations are known.

\subparagraph{Alternative definition}
The regular $\calD$-length of a monoid $\calM$
is the size of
its largest chain of regular $\calD$-classes.
A $\calD$-class of $\calM$
is an equivalence class of the preorder
$\leq_{\calD}$ defined by 
$m \leq_{\calD} m'$ if $m = s \cdot m' \cdot t$ for some $s,t \in \calM$,
and it is called regular
if it contains at least one idempotent element
(see~\cite{pin2010mathematical} for more details).
The equivalence between both definitions is proved in Appendix \ref{app:Green}.

\subparagraph{Computing the regular $\calD$-length}
The following table compares the size and the regular $\calD$-length
of the monoids mentioned earlier.
The entries corresponding to the sizes
are considered to be general knowledge.
We detail below the row listing the regular $\calD$-lengths.

\begin{center}
{
\renewcommand{\arraystretch}{1.5}

\begin{tabular}
{|l||c|c|c|c|c}
\hline
Monoid&
\makebox[1.5cm]{$G$}&
\makebox[1.5cm]{$\maxSemigroup{n}$}&
\makebox[1.5cm]{$\transformation{n}$}&
\makebox[1.5cm]{$\relation{n}$}
\\\hline\hline
Size
&$|G|$&$n$&$(n+1)^n$
&$2^{(n^2)}$\\\hline
Regular $\calD$-length  
&$1$&$n$&$n+1$
&$\frac{n^2+n+2}{2}$\\\hline
\end{tabular}
}
\end{center}

First, every group $\calG$ contains a single idempotent element
(the neutral element), hence its regular $\calD$-length is $1$.
Then, using the definition of the regular $\calD$-length
in terms of embedded max monoid,
we immediately obtain that
$\reglength{\maxSemigroup{n}}$ is equal to $n$.
We get the next entry
using the definition of the regular $\calD$-length
in terms of chain of $\calD$-classes:
The transformation monoid $\transformation{n}$
is composed of a single chain of $n+1$ $\calD$-classes
that are all regular~\cite{ganyushkin2008classical},
hence its regular $\calD$-length is $n+1$.

Finally, for the relation monoid $\relation{n}$,
the situation is not as clear:
the $\calD$-classes
do not form a single chain,
and some of them
are not regular.
Determining the exact size of the largest chain of $\calD$-classes
(note the absence of ``regular'')
is still an open question,
yet it is known to grow
exponentially with respect to $n$:
a chain of $\calD$-classes
whose size is the Fibonacci number $F_{n+3}-1$
is constructed in \cite{breen1991maximal},
and, conversely, the upper bound
$2^{n-1}+n-1$ is proved
in~\cite{konieczny1992cardinalities}
(and slightly improved
in~\cite{li1995konieczny,
zhang1999cardinalities,
hong2000distribution}).
Our second main result is that,
as long as we only consider chains of \emph{regular} $\calD$-classes,
we can obtain the precise value of the maximal length,
and, somewhat surprisingly,
it is only quadratic in $n$:
\begin{restatable}{theorem}{JLength}
\label{thm:JLength}
The regular $\calD$-length
of the monoid of $n \times n$ Boolean matrices is
$\frac{n^2+n+2}{2}$.
\end{restatable}
\noindent
Therefore, the regular $\calD$-length of
a transformation monoid
is exponentially
smaller than its size,
and the regular $\calD$-length of
a relation monoid
is even exponentially smaller than
its largest chain of $\calD$-classes.
For such kind of monoids,
Theorem \ref{theorem:general_valram}
performs considerably better than 
previously known methods
to find idempotent factors.
For instance, it was used in~\cite{MuschollP19} to close
the complexity gap left in~\cite{BaschenisGMP18}
for the problem of
deciding whether the function defined
by a given two-way word transducer
is definable by a one-way transducer.

\section{Definitions and notations}

We define in this section the
notions that are used throughout the
paper.
We denote by $\mathbb{N}$ the set
$\{0,1,2,\ldots\}$, 
and for all $i \leq j \in \mathbb{N}$
we denote by $[i,j]$ the interval
$\{i,i+1,\ldots,j\}$.

\subparagraph{Monoids}
A (finite) \emph{semigroup} $(\calS,\cdot)$
is a finite set $\calS$
equipped with a binary operation
$\cdot : \calS \times \calS \rightarrow \calS$
that is \emph{associative}:
$
(s_1 \cdot s_2) \cdot s_3 =
s_1 \cdot (s_2 \cdot s_3)
\textup{ for every }s_1,s_2,s_3 \in \calS. 
$
A \emph{monoid}
is a semigroup $(\calM,\cdot)$
that contains a \emph{neutral element} $1_{\calM}$:
$
m \cdot 1_\calM = m = 1_\calM \cdot m
\textup{ for all } m \in \calM.
$
A \emph{group}
is a monoid $(\calG,\cdot)$
in which every element $g \in \calG$
has an \emph{inverse element} $g^{-1} \in \calG$:
$
g \cdot g^{-1} = 1_{\calG} = g^{-1} \cdot g.
$
We always denote the semigroup operation
with the symbol $\cdot$.
As a consequence, we
identify a semigroup $(\calS,\cdot)$
with its set of elements $\calS$.

An element $e$ of a semigroup $\calS$
is called \emph{idempotent} if it satisfies
$e \cdot e = e$.
Note that whereas a finite semigroup
does not necessarily contain a neutral element,
it always contains at least one idempotent element:
iterating any element $s \in \calS$
eventually yields an idempotent element,
called the \emph{idempotent power}
of $s$,
and denoted $\idempower{s} \in \calS$.

A \emph{homomorphism} 
between two monoids $\calM$ and $\calM'$
is a function $\phi: \calM \rightarrow \calM'$
preserving the monoid structure:
$\phi(m_1 \cdot m_2) = \phi(m_1) \cdot \phi(m_2)$
for all $m_1,m_2 \in \calM$
and $\phi(1_{\calM}) = \phi(1_{\calM'})$.
A \emph{monomorphism} is an injective homomorphism,
an \emph{isomorphism} is a bijective homomorphism.

\subparagraph{Ramsey decomposition}
Let $\calM$ be a monoid.
A \emph{word} over $\calM$
is a finite sequence
$u = m_1 m_2 \ldots m_n \in \calM^*$
of elements of $\calM$.
The \emph{length} of $u$ is its number
of symbols $|u|=n \in \mathbb{N}$.
We enumerate the positions between the letters of $u$
starting from $0$ before the first letter,
until $|u|$ after the last letter.
A \emph{factor} of $u$ is a subsequence of $u$
composed of the letters between two such positions $i$ and $j$:
$u[i,j] = m_{i+1} m_{i+2} \ldots m_j \in \calM^*$
for some $0 \leq i \leq j \leq |u|$
(where $u[i,j] = \epsilon$ if $i = j$).
We denote by $\product{u}$ the element
$1_{\calM} \cdot m_1 \cdot m_2 \cdot \ldots \cdot m_n \in \calM$,
and we say that $u$ \emph{reduces} to $\product{u}$.
For every integer $\vari \in \mathbb{N}$,
a \emph{$\vari$-decomposition} of
$u$ is a decomposition of $u$ in
$\vari+2$ factors such that the $\vari$ middle 
ones are non-empty:
\[
u = x y_1 y_2 \ldots y_{\vari} z,
\textup{ where $x,z \in \calM^*$,
and $y_i \in \calM^+$
for every $1 \leq i \leq \vari$.}
\]
A $\vari$-decomposition is called
\emph{Ramsey}
if all the middle factors
$y_1,y_2, \ldots, y_{\vari}$
reduce to the same idempotent element
$e \in \calM$.
For instance,
a word has a Ramsey $1$-decomposition
if and only if it contains a factor that reduces
to an idempotent element.
The
\emph{Ramsey function}
$\truevalR{\calM} : \mathbb{N} \rightarrow \mathbb{N}$
associated to $\calM$
is the function mapping each
$\vari \in \mathbb{N}$
to the minimal
$\truevalRam{\calM}{\vari} \in \mathbb{N}$
such that every word $u \in \calM^*$
of length $\truevalRam{\calM}{\vari}$
has a Ramsey $\vari$-decomposition.

\section{Ramsey decompositions}\label{sec:ramDec}

In this section, we bound the Ramsey function $\truevalR{\calM}$
associated to a monoid $\calM$.
As a first step
we consider two basic cases
for which the exact value of the Ramsey function
is obtained:
in Subsection \ref{subsec::group_val}
we show that every group $\calG$ satisfies
$\truevalRam{\calG}{\vari} = \vari|\calG|$,
and in Subsection \ref{subsec::max_val}
we show that every
max monoid $\maxSemigroup{n}$
(obtained by equipping the first $n$ positive integers
with the $\max$ operation)
satisfies
$\truevalRam{\maxSemigroup{n}}{\vari} = \vari^n$.
Finally, in Subsection \ref{subsec::gen_val},
we prove bounds in the general case
by studying the submonoids
of $\calM$ isomorphic to a max monoid.

\subsection{Group: prefix sequence algorithm}\label{subsec::group_val}
We show that in a group,
the Ramsey function is polynomial with respect to
the size.
\begin{proposition}\label{proposition:group_valram}
For every group $\calG$,
$
\truevalRam{\calG}{\vari} =
\vari|\calG|
$
for all $\vari \in \mathbb{N}$.
\end{proposition}
We fix for this
subsection a group $\calG$ and $\vari \in \mathbb{N}$.
We begin by proving an auxiliary lemma,
which we then apply to prove
matching bounds
for $\truevalRam{\calG}{\vari}$:
First, we define an algorithm
that extracts a Ramsey $\vari$-decomposition
out of every word of length $\vari|\calG|$.
Then, we present the construction of a witness
$u_{\calG} \in \calG^*$ of length $\vari|\calG|-1$
that has no Ramsey $\vari$-decompositions.

\subparagraph{Key lemma}
In a group, the presence of inverse elements
allows us to establish a correspondence between
the factors of a word $u \in \calG^*$
that reduce to the neutral element,
and the pairs of prefixes of $u$
that both reduce to the same element.
\begin{lemma}\label{lemma::pref}
Two prefixes $u[0,i]$ and $u[0,j]$ of a word
$u \in \calG^*$
reduce to the same element
if and only if $u[i,j]$ reduces to the
neutral element of $\calG$.
\end{lemma}

\begin{proof}
Let $u \in \calG^*$ be a word.
The statement is a direct consequence of the fact that
for every $0 \leq i \leq j \leq |u|$,
$\product{u[0,i]} \cdot \product{u[i,j]} =
\product{u[0,j]}$:
If $\product{u[0,i]} = \product{u[0,j]}$,
then
\[
\product{u[i,j]} = 
\product{u[0,i]}^{-1} \cdot \product{u[0,j]} =
\product{u[0,i]}^{-1} \cdot \product{u[0,i]}
= \neutral{\calG}.
\]
Conversely, if $\product{u[i,j]} = \neutral{\calG}$,
then
\[
\product{u[0,i]} =
\product{u[0,i]} \cdot \neutral{\calG} =
\product{u[0,i]} \cdot \product{u[i,j]} =
\product{u[0,j]}.
\qedhere
\]
\end{proof}

\subparagraph{Algorithm}
We define an algorithm
constructing Ramsey $\vari$-decompositions.
\begin{enumerate}[]
    {\setlength\itemindent{-15pt}
    \item[]\textsc{Alg}$_1$: Start with
    $u \in \calG^*$
    of length $\vari|\calG|$;}
    \begin{enumerate}[nolistsep]
    {\setlength\itemindent{-10pt}
    \item 
    Compute the $\vari|\calG|+1$ prefixes
    $\product{u[0,0]}$, $\product{u[0,1]}$, 
    \ldots, $\product{u[0,|u|]}$
    of $u$;
    }
    {\setlength\itemindent{-10pt}
    \item
    Find $\vari+1$ indices
    $i_0$, $i_1$, \ldots, $i_{\vari}$
    such that all the $\product{u[0,i_j]}$ are equal;
    }
    {\setlength\itemindent{-10pt}
    \item
    Return the Ramsey $\vari$-decomposition
    $u = u[i_0,i_1]u[i_1,i_2] \ldots
    u[i_{\vari-1},i_{\vari}]$.
    }
    \end{enumerate}
\end{enumerate}
Since Lemma \ref{lemma::pref}
ensures that every pair of elements
$i_{j}$, $i_{j+1}$ identified
at step $2$ satisfies
$\product{u[i_j,i_{j+1}]} = \neutral{\calG}$,
we are guaranteed that the returned
$\vari$-decomposition is Ramsey.

\subparagraph{Witness}
We build a word $u_{\calG} \in \calG^*$
of length $\vari|\calG|-1$
that has no Ramsey $\vari$-decompositions.
Let $v = a_1a_2 \ldots a_{\vari |\calG|} \in \calG^*$
be a word of length
$\vari |\calG|$,
starting with the letter $\neutral{\calG}$,
and containing exactly $\vari$ times
each element of $\calG$.
For instance, given an enumeration
$g_1,g_2, \ldots, g_{|\calG|}$
of the elements of $\calG$
starting with $g_1 = \neutral{\calG}$,
we can simply pick
$
v=g_1^{\vari}g_2^{\vari} \ldots
g_{|\calG|}^{\vari}.
$
Now let
$u_{\calG} = b_1b_2 \ldots b_{\vari |\calG|-1}$
be the word
whose sequence of reduced prefixes is $v$:
for every $1 \leq i \leq \vari |\calG|-1$,
the letter $b_i$ is equal to $a_{i}^{-1} \cdot a_{i+1}$.
Then for every $\vari$-decomposition of $u_{\calG}$,
at least one of the factors do not reduce
to the neutral element of $\calG$,
since otherwise
Lemma \ref{lemma::pref}
would imply the existence of $\vari + 1$
identical letters in $v$,
which is not possible by construction.
As a consequence,
$u_{\calG}$ has no Ramsey $\vari$-decompositions.

\subsection{Max monoid: divide and conquer algorithm}\label{subsec::max_val}

Given an integer $n \in \mathbb{N}$,
the \emph{max monoid}, denoted $\maxSemigroup{n}$,
is the monoid over the set $\{1,2,\ldots,n\}$
with the associative operation
$i \cdot j = \max(i,j)$.
Whereas in a group
only the neutral element is idempotent,
each element $i$ of the max monoid
$\maxSemigroup{n}$
is idempotent since $\max(i,i) = i$.
As a result of this abundance of idempotent elements,
an exponential bound is required to ensure
the presence of consecutive factors
reducing to the same idempotent element.

\begin{proposition}\label{proposition:maxsemigroup_valram}
For every max monoid $\maxSemigroup{n}$,
$
\truevalRam{\maxSemigroup{n}}{\vari} =
\vari^{n}
$
for all $\vari \in \mathbb{N}$.
\end{proposition}
The proof is done in two steps:
we first define an algorithm
that extracts a Ramsey $\vari$-decomposition
out of every word of length $\vari^{n}$,
and then
we present the construction of a witness
$u_{n}$ of length $\vari^{n}-1$
that has no Ramsey $\vari$-decompositions.

\subparagraph{Algorithm}
We define an algorithm that extracts
a Ramsey $\vari$-decompositions out of each word
$u \in \maxSemigroup{n}^*$
of length $\vari^n$.
It is a basic divide and conquer algorithm:
we divide the initial word $u$
into $\vari$ equal parts.
If each of the $\vari$ parts reduces to $n$,
they form a Ramsey $\vari$-decomposition
since $n$ is an idempotent element.
Otherwise, one part does not contain the maximal element
$n \in \maxSemigroup{n}$, and we start over with it.
Formally,
\begin{enumerate}[]
    {\setlength\itemindent{-15pt}
    \item[]\textsc{Alg}$_2$: Start with
    $u \in \maxSemigroup{n}^*$
    of length $\vari^n$, initialize $j$ to $n$.
    While $j>0$, repeat the following:}
    \begin{enumerate}[nolistsep]
    {\setlength\itemindent{-10pt}
    \item
    Split $u$ into $\vari$ factors
    $u_1$, $u_2$, \ldots, $u_{\vari}$
    of length $\vari^{j-1}$;
    }
    {\setlength\itemindent{-10pt}
    \item\label{algo::max_stop}
    If every $u_i$ contains the letter $j$,
    return the Ramsey $\vari$-decomposition
    $u=u_1 u_2 \ldots u_\vari$;
    }
    {\setlength\itemindent{-10pt}
    \item
    If $u_i$ does not contain $j$ for some $1 \leq i \leq j$,
    decrement $j$ by $1$ and
    set $u \coloneqq u_i \in \maxSemigroup{j-1}^*$.
    }
    \end{enumerate}
\end{enumerate}
The algorithm is guaranteed
to eventually return
a Ramsey $\vari$-decomposition:
if the $n^{\textup{th}}$ cycle of the algorithm
is reached,
it starts with a word of length $k$ whose letters are in
the monoid $\maxSemigroup{1}$,
which only contains the letter $1$,
hence the algorithm will go to step b.

\subparagraph{Witness}
We construct an infinite sequence
of words
$u_1,u_2, \ldots \in \mathbb{N}^*$
such that for all $n \in \mathbb{N}$,
\begin{enumerate*}[(a)]
\item
$u_n \in \maxSemigroup{n}$ satisfies $|u_n| = \vari^n-1$ and
\item
$u_n$ has no Ramsey $\vari$-decompositions.
\end{enumerate*}
Let
\[
\begin{array}{llll}
u_1 & = & 1^{\vari-1} \in \maxSemigroup{1}^*,\\
u_{n} & = & (u_{n-1}n)^{\vari-1}u_{n-1}
\in \maxSemigroup{n}^*
& \textup{for every } n > 1.
\end{array}
\]
For every $n>1$, the word $u_n$ is defined as $\vari$
copies of $u_{n-1}$ separated by the letter $n$.
We prove by induction
that the two conditions are satisfied
by each word of the sequence.
The base case is immediate:
the word $u_1$ has length $\vari-1$,
and as a consequence has
no decomposition into 
$\vari$ nonempty factors.
Now suppose that $n>1$,
and that $u_{n-1}$
satisfies the two properties.
Then $u_n$ has the required length:
\[
|u_{n}| = (\vari-1)(|u_{n-1}| + 1) + |u_{n-1}|
=(\vari-1)\vari^{n-1} + \vari^{n-1}-1
= \vari^n -1.
\]
To conclude, we show that every 
$\vari$-decomposition
\begin{equation}\label{equation:decompG}
u_n = xy_1y_2 \ldots y_{\vari}z,
\textup{ with } y_i \in \maxSemigroup{n}^+
\textup{ for all } 1 \leq i \leq \vari
\end{equation}
is not Ramsey.
Let $y$ be the factor $y_1y_2 \ldots y_{\vari}$
of $u_n$, and consider the two following cases:
\begin{itemize}
\item
If $\product{y} \neq n$,
none of the $y_i$ contains the letter $n$,
hence $y$ is factor of one of the factors
$u_{n-1}$ of $u_{n}$.
Therefore, by the induction hypothesis,
Decomposition (\ref{equation:decompG}) is not Ramsey.
\item 
If $\product{y} = n$,
since $u_{n}$ contains only 
$\vari-1$ copies of the letter $n$,
one of the factors $y_i$ does not contain $n$
for $1 \leq i \leq \vari$.
Then $\product{y} \neq \product{y_i}$, hence Decomposition (\ref{equation:decompG}) is not Ramsey.
\end{itemize}

\begin{example}
Here are the first three words of the sequence
in the cases $\vari=2$ and $\vari=3$:
\[
\begin{array}{lllll}
\vari = 2: & u_1 = 1 & u_2 = 121 & u_3 = 1213121,\\
\vari = 3: & u_1 = 11 & u_2 = 11211211 & u_3 = 11211211311211211311211211.
\end{array}
\]
\end{example}

\subsection{General setting}\label{subsec::gen_val}
We saw in the previous subsection that
for the max monoid $\maxSemigroup{n}$,
words of length exponential
with respect to $n$ are required
to guarantee the presence of Ramsey decompositions
(Proposition \ref{proposition:maxsemigroup_valram}).
Note that the same lower bound applies to
every monoid $\calM$ that contains
a copy of $\maxSemigroup{n}$
as submonoid.
We now show that we can also obtain
an upper bound for $\truevalRam{\calM}{\vari}$
by studying the submonoids of $\calM$
isomorphic to a max monoid.
We formalise this idea through the notion of
regular $\calD$-length of a monoid.

\subparagraph{Regular $\calD$-length}
The \emph{regular $\calD$-length} of a monoid $\calM$,
denoted $\reglength{\calM}$,
is the size
of the largest max monoid embedded in $\calM$.
Formally, it is the largest $\ell \in \mathbb{N}$
such that there exists a monomorphism
(i.e. injective monoid homomorphism)
$\phi: \maxSemigroup{\ell} \rightarrow \calM$.
We now present the main theorem of this section,
which states that for every monoid $\calM$,
the degree of
$\truevalRam{\calM}{\vari}$
is determined by the regular $\calD$-length of $\calM$.

\ThmValRam*
Let us fix for the whole subsection a monoid $\calM$
of regular $\calD$-length $\reglength{\calM}$
and an integer $\vari \in \mathbb{N}$.
The lower bound is a corollary of
Proposition \ref{proposition:maxsemigroup_valram}:
the max monoid $\maxSemigroup{\reglength{\calM}}$
has a witness $u_{\reglength{\calM}}$ of length
$\vari^{\reglength{\calM}}-1$ that has no Ramsey
$\vari$-decompositions
(its construction is presented
in the previous subsection).
Then, by definition of the regular $\calD$-length,
there exists a monomorphism
$\phi: \maxSemigroup{\reglength{\calM}} \rightarrow \calM$,
and applying $\phi$ to $u_{\reglength{\calM}}$ letter by letter
yields a witness $u_{\reglength{\calM}}' \in \calM^*$
of length $\vari^{\reglength{\calM}}-1$ that has no Ramsey
$\vari$-decompositions.

The rest of the subsection is devoted to
the proof of the upper bound.
We begin by defining an auxiliary algorithm
that extracts from each long enough word
a decomposition where the prefix
and suffix absorb the middle factors.
Then, we define our main algorithm which,
on input $u \in \calM^*$
of length $(\vari|\calM|^{4})^{n}$ for some
$n \in \mathbb{N}$,
either returns a Ramsey
$\vari$-decomposition
of $u$,
or a copy of the max monoid $\maxSemigroup{n+1}$
embedded in $\calM$.
In particular, if $n$ is equal to the
regular $\calD$-length $\reglength{\calM}$ of $\calM$,
we are guaranteed to obtain a Ramsey $\vari$-decomposition.

\subparagraph{Auxiliary algorithm}
We define an algorithm which,
on input $u \in \calM^*$
of length $\vari |\calM|^2$,
returns a $\vari$-decomposition
\[
u = x y_1 y_2 \ldots y_{\vari} z,
\textup{ where $x,z \in \calM^*$,
and $y_i \in \calS^+$
for every $1 \leq i \leq \vari$}
\]
such that for every $1 \leq i \leq \vari$,
both $x$ and $z$ are able to absorb the factor
$y_i$:
$\product{xy_i} = \product{x}$ and 
$\product{y_iz} = \product{z}$.
This is done as follows:
since $u$ is a word of length $\vari|\calM|^2$,
it can be split into $\vari|\calM|^2 + 1$
distinct prefix-suffix pairs.
Then $\vari+1$ of these pairs
reduce to the same pair of elements of $\calM$,
which immediately yields the desired decomposition.
Formally,
\begin{enumerate}[]
    {\setlength\itemindent{-15pt}
    \item[]\textsc{Alg}$_3$: Start with
    $u \in \calM^*$
    of length $\vari|\calM|^{2}$;}
    \item
    \begin{enumerate}
        \item
        Compute the $\vari|\calM|^{2}+1$ prefixes
        $\product{u[0,0]}$, $\product{u[0,1]}$, 
        \ldots, $\product{u[0,|u|]} \in \calM$
        of $u$,
        \item
        Compute the $\vari|\calM|^{2}+1$ suffixes
        $\product{u[0,|u|]}$, $\product{u[1,|u|]}$,
        \ldots, $\product{u[|u|,|u|]} \in \calM$
        of $u$,
        \item
        Identify $\vari+1$ indices
        $s_0$, $s_1$, \ldots, $s_{\vari}$
        such that
        \begin{enumerate*}[(1)]
        \item
        all the $\product{u[0,s_i]}$ are equal, 
        \item
        all the $\product{u[s_i,|u|]}$ are equal;
        \end{enumerate*}
    \end{enumerate}
    \item
    Set $x = u[0,s_0]$,
    $z = u[s_\vari,|u|]$,
    and $y_i = u[s_{i-1},s_i]$
    for every $1 \leq i \leq \vari$;
    \item
    Return the $\vari$-decomposition
    $x y_1 y_2 \ldots y_{\vari} z$
    of $u$.
\end{enumerate}

\subparagraph{Main algorithm}
We define an algorithm
extracting Ramsey $\vari$-decompositions.
Over an input $u \in \calM^*$ of length $(\vari|\calM|^{4})^{n}$ for $n \in \mathbb{N}$,
the algorithm works by defining
gradually shorter words
$u_n,u_{n-1},\ldots \in \calM^*$,
where each $u_j$ has length $(\vari|\calM|^{4})^{j}$,
along with a sequence
of idempotent elements $e_{n+1},e_{n}, \ldots \in \calM$.
Starting with $u_n = u$,
we define $e_{n+1}$ as the idempotent power
of some well chosen factors of $u_n$.
We then consider $\vari$ consecutive factors of $u_n$.
If all of them reduce to $e_{n+1}$,
they form a Ramsey $\vari$-decomposition, and we are done.
Otherwise, we pick a factor $u_{n-1}$ that
does not reduce to $e_{n+1}$,
and we start over.
This continues until either a Ramsey $\vari$-decomposition
is found, or $n$ cycles are completed.
In the later case,
we show that the function
$\phi : \maxSemigroup{n+1}  \rightarrow \calM$
mapping $i$ to $e_i$ is a monomorphism.

\begin{itemize}[]
    {\setlength\itemindent{-15pt}
    \item[] \textsc{Alg}$_4$: Start with
    $u \in \calM^*$
    of length $(\vari|\calM|^{4})^{n}$.
    Initialize $u_n$ to $u$ and $j$ to $n$.}
    \item 
    While $j>0$, repeat the following:
    \begin{enumerate}[]
    \item\label{algo::gen_first}
    \begin{enumerate}
        \item
        Call \textsc{Alg}$_3$ to get an $m$-decomposition
        $u_j = x y_1 y_2 \ldots y_m z$,
        where $m = \vari^j|\calM|^{4j-2}$;
        \item
        Set
        $v \coloneqq \product{y_1}\product{y_2}
        \ldots \product{y_m} \in \calM^*$;
    \end{enumerate}
    \item\label{algo::gen_second}
    \begin{enumerate}
        \item
        Call \textsc{Alg}$_3$ to get an $m'$-decomposition
        $v = x' y_1' y_2'\ldots y_{m'}'z'$,
        where $m' = \vari^j|\calM|^{4j-4}$;
        \item\label{algo::gen_def_e}
        Set
        $w \coloneqq \product{y_1'}\product{y_2'}
        \ldots \product{y_{m'}'} \in \calM^*$,
        and set
        $e_{j+1} \coloneqq (\product{z'x'})^{\#}$;
    \end{enumerate}
    \item\label{algo::gen_third}
    \begin{enumerate}
        \item
        Split $w$ into $\vari$ factors
        $y_1''$, $y_2''$, \ldots, $y_{\vari}''$
        of length $(\vari|\calM|^{4})^{j-1}$;
        \item\label{algo::gen_stop}
        If every $y_i''$ satisfies
        $\product{y_i} = e_{j+1}$,
        then $w=y_1''y_2''\ldots y_\vari''$
        is a Ramsey decomposition.
        Return the corresponding Ramsey
        $\vari$-decomposition of $u$;
        \item\label{algo::gen_def_u}
        If $\product{y_i''} \neq e_{j+1}$
        for some $1 \leq i \leq n$,
        set $u_{j-1} \coloneqq y_i''$,
        and decrement $j$ by $1$.
    \end{enumerate}
    \end{enumerate}
    \item
    Set $e_1=1_{\calM}$, and return the idempotent elements
    $e_1,e_2, \ldots, e_{n+1} \in \calM$.
\end{itemize}
\textbf{Step \ref{algo::gen_first}.}
We use the auxiliary algorithm to obtain
a decomposition $u_j = xy_1y_2 \ldots y_m z$,
and we build $v$ by concatenating
the reductions of the $y_i$.
Since both $x$ and $z$
absorb each $y_i$,
and in step \ref{algo::gen_def_e}
we define $e_{j+1}$ as
the idempotent power of reduced factors of $v$:
\begin{equation}\label{equ::uJe}
\textup{The word $u_j$,
its prefix $x$
and its suffix $z$ satisfy }
\product{u_j} = \product{xz} = \product{x} \cdot e_{j+1} \cdot \product{z}.
\end{equation}
\textbf{Step \ref{algo::gen_second}.}
We use the auxiliary algorithm
to get a decomposition
$u' = x'y'_1y'_2 \ldots y'_{m'} z'$,
we build $w$ by concatenating
the reductions of the $y_i'$,
and we set $e_{j+1}$ as
the idempotent power of $\product{z'x'}$.
As both $x'$ and $z'$
absorb each $y'_i$,
and in step \ref{algo::gen_def_u}
we define $u_{j-1}$ as
a factor of $w$:
\begin{equation}\label{equ::eHu}
\textup{For every factor $y$ of $u_{j-1}$, }
e_{j+1} \cdot \product{y}
= e_{j+1}
= \product{y} \cdot e_{j+1}.
\end{equation}
\textbf{Step \ref{algo::gen_third}.}
We divide $w$ into $\vari$ factors of equal length.
If each of them reduces to $e_{j+1}$,
they form a Ramsey $\vari$-decomposition
of $w$.
As $w$ is obtained form $u$
by iteratively reducing factors
and dropping prefixes and suffixes,
this decomposition can be transferred back
to a Ramsey $\vari$-decomposition of $u = u_n$.
If one factor does not reduce to $e_{j+1}$,
we assign its value to $u_{j-1}$.
Therefore:
\begin{equation}\label{equ::uNe}
\textup{
The word $u_{j-1}$
does not reduce to
$e_{j+1}$.}
\end{equation}

\subparagraph{Proof of correctness}
To prove that the algorithm behaves as intended,
we show that if it completes
$n$ cycles without returning a
Ramsey $\vari$-decomposition,
then the function
$\phi: \maxSemigroup{n+1} \rightarrow \calM$
defined by $\phi(j) = e_j$ 
is a monomorphism.
Since $e_{j}$ is the idempotent power of
reduced factors of $u_{j-1}$
for all $1 \leq j \leq n$,
Equation (\ref{equ::eHu})
yield that
$e_{j+1}\cdot e_{j}
=e_{j+1}
=e_{j} \cdot e_{j+1}$.
Therefore $\phi$ is a homomorphism.
We conclude by showing that it is injective.
Suppose, towards building a contradiction,
that $\phi(j) = e_{j} = e_{i} = \phi(i)$
for some $1 \leq j<i \leq n$.
Since $\phi$ is a homomorphism,
all the intermediate elements collapse:
in particular $e_{j} = e_{j+1}$.
Then
\[
\product{u_{j-1}}
\underset{(\ref{equ::uJe})}{=}
\product{x}\cdot e_{j}\cdot \product{z}
= \product{x}\cdot e_{j+1}\cdot \product{z}
\underset{(\ref{equ::eHu})}{=}
e_{j+1},
\]
which cannot hold by
Equation (\ref{equ::uNe}).

\section{Regular $\calD$-length of the
monoid of Boolean matrices}\label{sec:reglength}

A Boolean matrix
is a matrix $A$ whose components
are Boolean elements: $A_{ij} \in \{0,1\}$.
The (full) \emph{Boolean matrix monoid}
$\relation{n}$ is the set of all $n \times n$
Boolean matrices,
equipped with the matrix composition defined as follows:
$(A \cdot B)_{ik} = 1$
if and only if there exists $j \in \firstintegers{n}$
satisfying
$A_{ij}=B_{jk}=1$.
This fits the standard matrix multiplication
if we consider that $1+1=1$:
addition of Boolean elements
is the OR operation,
and multiplication is the AND operation.
The main contribution of this section is
the following theorem.
\JLength*
\noindent
The proof is split in two parts.
We prove the upper bound
by studying the structure of the
idempotent elements of $\relation{n}$
(Subsection \ref{subsec::up_relation}).
Then, we prove the lower bound
by constructing a monomorphism
from the max monoid of size $\frac{n^2+n+2}{2}$
into $\relation{n}$
(Subsection \ref{subsec::low_relation}).
We begin by introducing definitions tailored
to help us in the following demonstrations.


\subparagraph{Stable matrix}
A Boolean matrix $A \in \relation{n}$
is called \emph{stable}
if for each component
$A_{ik}$ equal to $1$,
there exists $j \in \firstintegers{n}$
satisfying $A_{ij} = A_{jj} = A_{jk} = 1$.
Idempotent matrices are stable (Appendix \ref{app:mat}).





\subparagraph{Positive set}
A (maximal) \emph{positive set} of
an idempotent matrix
$A \in \relation{n}$
is a maximal set
$I \subseteq \firstintegers{n}$
such that all the corresponding components of $A$
are $1$:
$A_{ij} =1$ for all $i,j \in I$,
and 
for every $k \in \firstintegers{n} \setminus I$,
there exists $i \in I$
such that $A_{ik} = 0$ or $A_{ki}=0$.
The positive sets of an idempotent matrix
are disjoint (Appendix \ref{app:mat}),
hence $A$ has at most $n$
positive sets.

\subparagraph{Free pair}
For each idempotent matrix $A \in \relation{n}$
we define the relation $\refrel_A$ on
$\firstintegers{n}$ as follows:
given
$i,j \in \firstintegers{n}$,
we have $i \refrel_A j$
if for all 
$i_2,j_2 \in \firstintegers{n}$,
$A_{i_2i} = 1 = A_{jj_2}$ implies $A_{i_2j_2} = 1$.
A \emph{free pair} of $A$
is a set of two distinct elements
$i,j \in \firstintegers{n}$
incomparable by $\refrel_A$:
$i\not\refrel_A j$
and $j\not\refrel_A i$.
Note that $A$ has at most $\frac{n(n-1)}{2}$
free pairs (all sets of two distinct elements in $\firstintegers{n}$).
Let us state some observations concerning
$\refrel_A$ that follow immediately
from the definition.
First, as $A$ is idempotent,
$\refrel_A$ is reflexive (Appendix \ref{app:mat}).
However, it might not be transitive.
Moreover, for every component $A_{ij}$ of $A$
equal to $1$,
we have that $i\refrel_A j$ (Appendix \ref{app:mat}).
The converse implication is not true,
as shown by the following example.
Finally, for every $i \in \firstintegers{n}$,
if the $i^{\textup{th}}$ row contains
no $1$,
i.e., $A_{ik} = 0$
for all $k \in \firstintegers{n}$,
then $i\refrel_A j$ for every $j \in \firstintegers{n}$.
Conversely,
if the $i^{\textup{th}}$ column contains
no $1$,
then $j\refrel_A i$
for every $j \in \firstintegers{n}$.

\subparagraph{Example}
We depict below a submonoid of $\relation{4}$
generated by two matrices $A$ and $B$.
The six elements of this submonoid,
including the identity matrix $D \in \relation{n}$,
are all idempotent.
Under each matrix, we list its positive sets.
We then compute the corresponding free pairs.

  \begin{center}
    \begin{tikzpicture}
      \foreach \c in {-1,0,1,2,3,4}{
        \pgfmathtruncatemacro{\stex}{\c}
        \coordinate (i\c) at
            (2.35*\stex+0.375,-1.1);
        \foreach \d in {0,1,2,3}{
            \foreach \e in {0,1,2,3}{
                \coordinate (a\c\d\e) at
                (2.35*\c + 0.25*\d, -0.25*\e);
                \ifthenelse{\c>-1}{
                \ifthenelse{\e=1 \and \d=0}
                {\node
                [draw,thick,circle,inner sep=1]
                (m\c\d\e) at (a\c\d\e)
                {};}{}
                \ifthenelse{\e=1 \and \d=2}
                {\node
                [draw,thick,circle,inner sep=1]
                (m\c\d\e) at (a\c\d\e)
                {};}{}
                \ifthenelse{\e=3 \and \d=0}
                {\node
                [draw,thick,circle,inner sep=1]
                (m\c\d\e) at (a\c\d\e)
                {};}{}
                \ifthenelse{\e=3 \and \d=2}
                {\node
                [draw,thick,circle,inner sep=1]
                (m\c\d\e) at (a\c\d\e)
                {};}{}
                \ifthenelse{\e=0 \and \d=0}
                {\draw[thick] ($(a\c\d\e)+(0,0.06)$)--($(a\c\d\e)-(0,0.06)$);}{}
                \ifthenelse{\e=3 \and \d=3}
                {\draw[thick] ($(a\c\d\e)+(0,0.06)$)--($(a\c\d\e)-(0,0.06)$);}{} 
                \ifthenelse{\e=0 \and \d=2}
                {\draw[thick] ($(a\c\d\e)+(0,0.06)$)--($(a\c\d\e)-(0,0.06)$);}{} 
                \ifthenelse{\e=1 \and \d=3}
                {\draw[thick] ($(a\c\d\e)+(0,0.06)$)--($(a\c\d\e)-(0,0.06)$);}{} 
                \ifthenelse{\c>0 \and \e=0 \and \d=1}
                {\draw[thick] ($(a\c\d\e)+(0,0.06)$)--($(a\c\d\e)-(0,0.06)$);}{} 
                \ifthenelse{\c>0 \and \e=0 \and \d=3}
                {\draw[thick] ($(a\c\d\e)+(0,0.06)$)--($(a\c\d\e)-(0,0.06)$);}{} 
                \ifthenelse{\c>0 \and \e=2 \and \d=3}
                {\draw[thick] ($(a\c\d\e)+(0,0.06)$)--($(a\c\d\e)-(0,0.06)$);}{} 
                }
                {
                \ifthenelse{\e=\d}
                {\draw[thick] ($(a\c\d\e)+(0,0.06)$)--($(a\c\d\e)-(0,0.06)$);}
                {\node
                [draw,thick,circle,inner sep=1]
                (m\c\d\e) at (a\c\d\e)
                {};}
                }
            }
        }
        \ifthenelse{\c>-1}{
        \ifthenelse{\c=1 \OR \c=3}
        {\node[draw,thick,circle,inner sep=1]
        (m\c02) at (a\c02)
        {};
        \node[draw,thick,circle,inner sep=1]
        (m\c22) at (a\c22)
        {};}
        {\draw[thick] ($(a\c02)+(0,0.06)$)--($(a\c02)-(0,0.06)$);
        \draw[thick] ($(a\c22)+(0,0.06)$)--($(a\c22)-(0,0.06)$);}
        
        \ifthenelse{\c=1 \OR \c=2}
        {\node[draw,thick,circle,inner sep=1]
        (m\c11) at (a\c11)
        {};
        \node[draw,thick,circle,inner sep=1]
        (m\c13) at (a\c13)
        {};}
        {\draw[thick] ($(a\c11)+(0,0.06)$)--($(a\c11)-(0,0.06)$);
        \draw[thick] ($(a\c13)+(0,0.06)$)--($(a\c13)-(0,0.06)$);}
        
        \ifthenelse{\c=0}
        {\node[draw,thick,circle,inner sep=1]
        (m\c10) at (a\c10)
        {};
        \node[draw,thick,circle,inner sep=1]
        (m\c30) at (a\c30)
        {};
        \node[draw,thick,circle,inner sep=1]
        (m\c32) at (a\c32)
        {};}
        {}
        
        \ifthenelse{\c=0 \OR \c=1}
        {\node[draw,thick,circle,inner sep=1]
        (m\c12) at (a\c12)
        {};}
        {\draw[thick] ($(a\c12)+(0,0.06)$)--($(a\c12)-(0,0.06)$);}
        }{}
      }
     \node[](j-1) at (i-1) {\scriptsize $D$};
     \node[](k-1) at ($(i-1)-(0,0.35)$)
     {\scriptsize $\{1\},\{2\},\{3\},\{4\}$};
     \node[](j0) at (i0) {\scriptsize $A$};
     \node[](k0) at ($(i0)-(0,0.35)$)
     {\scriptsize $\{1,3\},\{2,4\}$};
     \node[](j1) at (i1) {\scriptsize $B$};
     \node[](k1) at ($(i1)-(0,0.35)$)
     {\scriptsize $\{1\},\{4\}$};
     \node[](j2) at (i2) {\scriptsize $A \cdot B$};
     \node[](k2) at ($(i2)-(0,0.35)$)
     {\scriptsize $\{1,3\},\{4\}$};
     \node[](j3) at (i3) {\scriptsize $B \cdot A$};
     \node[](k3) at ($(i3)-(0,0.35)$)
     {\scriptsize $\{1\},\{2,4\}$};
     \node[](j4) at (i4)
     {\scriptsize $A \cdot B \cdot A$};
     \node[](k4) at ($(i4)-(0,0.35)$)
     {\scriptsize $\{1,3\},\{2,4\}$};
    \end{tikzpicture}
  \end{center}
Every pair is free in $D$
since the relation $\refrel_{D}$ is the identity:
given two distinct elements $i,j \in \firstintegers{n}$,
we have $D_{ii} = 1 = D_{jj}$,
yet $D_{ij} = 0$, hence $i \not\refrel_{D} j$.
On the contrary,
the four matrices $B$, $A \cdot B$, $B \cdot A$
and $A \cdot B \cdot A$ has no free pairs:
the relation $\refrel_B$ only lacks $(4,1)$,
$\refrel_{A \cdot B}$ only lacks $(4,1)$ and $(4,3)$,
$\refrel_{B \cdot A}$ only lacks $(2,1)$ and $(4,1)$,
$\refrel_{A \cdot B \cdot A}$ only lacks
$(2,1)$, $(4,1)$ and $(4,3)$.
Finally, for $A$,
the relation $\refrel_A$ is the union of
the identity and the four pairs
$\{(1,3),(3,1),(2,4),(4,2)\}$,
which yields the free pairs
$\{1,2\}$, $\{1,4\}$, $\{2,3\}$ and $\{3,4\}$.

\subsection{Upper bound}
\label{subsec::up_relation}
To prove the upper bound of
Theorem \ref{thm:JLength},
we show that every monomorphism
$\phi: \maxSemigroup{m} \rightarrow \relation{n}$
satisfies
$m \leq \frac{n^2+n+2}{2}$.
To this end,
we study the sequence of matrices
$s_{\phi}=A_1,A_2, \ldots, A_m$
obtained by listing the elements
$\phi(i)=A_i$ of the image of $\phi$.
Note that
all the elements of $s_{\phi}$ are distinct
as $\phi$ is injective,
and $A_i \cdot A_{i+1} = A_{i+1} = A_{i+1} \cdot A_i$
for all $1 \leq i < m$
as $\phi$ is a homomorphism.
We introduce three lemmas 
that imply interesting properties
of every pair $A_i,A_{i+1}$
of successive matrices of $s_{\phi}$.
First, Lemma~\ref{lemma:clique_containment}
shows that
every positive set of $A_{i+1}$ contains a positive
set of $A_i$.
Therefore, since positive sets are disjoint,
the number of positive sets can
never increase along $s_{\phi}$.
Second, Lemma~\ref{lemma:incomparable_containment}
shows that
every free pair of $A_{i+1}$ is also
a free pair of $A_{i}$.
As a consequence, the number of free pairs can
never increase along $s_{\phi}$.
Finally, Lemma~\ref{prop::upper_bound}
shows that either the number of positive sets
or free pairs differs between
$A_i$ and $A_{i+1}$,
as otherwise these two matrices would be equal.

Combining the three lemmas
yields that between each pair of successive
matrices of $s_{\phi}$,
neither the number of positive sets
nor the number of free pairs increases,
and at least one decreases.
This immediately
implies the desired upper bound:
as the number of positive sets 
of matrices of $\relation{n}$
ranges from $0$ to $n$
and the number of free pairs ranges from
$0$ to $\frac{n(n-1)}{2}$,
$s_\phi$ contains at most
$n+\frac{n(n-1)}{2}+1 = \frac{n^2+n+2}{2}$
matrices.
To conclude, we now proceed with the formal statements
and the proofs of the three lemmas.

\begin{lemma}\label{lemma:clique_containment}
Let $A$ and $B$
be two idempotent matrices of $\relation{n}$
satisfying $A \cdot B= B = B \cdot A$.
Then every positive set of $B$
contains a positive set of $A$.
\end{lemma}

\begin{proof}
Let us pick two idempotent matrices
$A,B \in \relation{n}$
satisfying $A \cdot B= B = B \cdot A$.
If $B$ has no positive sets,
the statement is trivially satisfied.
Now let us suppose that $B$ has
at least one positive set
$I \subseteq \firstintegers{n}$.
We show the existence of
a positive set $J \subseteq I$ of $A$.

Since $I$ is not empty by definition,
it contains an element $i$,
and $B_{ii}=1$.
Then, as $B = B \cdot A$,
there exists $k \in \firstintegers{n}$
satisfying $B_{ik} = A_{ki} =1$.
Moreover, as $A$ is stable,
there exists $j \in \firstintegers{n}$
satisfying $A_{kj} = A_{jj} =A_{ji} =1$.
In particular, $A_{jj}=1$,
hence $A$ has a positive set
$J$ containing $j$.
Then, for every $i_2 \in I$
and every $j_2,j_3 \in J$,
we obtain
\[
\begin{array}{ll}
B_{j_2i_2} = (A \cdot A \cdot B)_{j_2i_2} = 1 
\textup{ since }
A_{j_2j}=A_{ji}=B_{ii_2}
=
1,\\
B_{i_2j_3}=(B\cdot B \cdot A \cdot A)_{i_2j_3}=1
\textup{ since }
B_{i_2i}=B_{ik}=A_{kj}=A_{jj_3}
=
1,\\
B_{j_2j_3}=(B \cdot B)_{j_2j_3}=1
\textup{ since }
B_{j_2i_2}=B_{i_2j_3}
=
1.
\end{array}
\]
As a consequence,
$J$ is a subset of $I$
since positive sets are maximal by definition.
\end{proof}

\begin{lemma}\label{lemma:incomparable_containment}
Let $A$ and $B$
be two idempotent matrices of $\relation{n}$
satisfying $A \cdot B= B = B \cdot A$.
Then every free pair of $B$
is a free pair of $A$.
\end{lemma}

\begin{proof}
Let us pick two idempotent matrices
$A,B \in \relation{n}$
satisfying $A \cdot B= B = B \cdot A$.
We prove the lemma by contraposition:
we show that for every pair of elements
$i,j \in \firstintegers{n}$,
$i \refrel_A j$ implies $i \refrel_B j$
(hence if $i$ and $j$ are incomparable
by $\refrel_B$, so are they by $\refrel_A$).

Let us pick $i,j \in \firstintegers{n}$
satisfying $i \refrel_A j$,
and $i_2,j_2 \in \firstintegers{n}$ 
satisfying
$B_{i_2i} = 1 = B_{jj_2}$.
To conclude,
we show that
$B_{i_2j_2}=1$.
To this end, we introduce two new elements
$i_1,j_1 \in \firstintegers{n}$:
First, as $(B \cdot A)_{i_2i}=B_{i_2i}=1$,
there exists $i_1\in \firstintegers{n}$
such that $B_{i_2i_1}=1$
and $A_{i_1i}=1$;
Second,
as $(A\cdot B)_{jj_2} = B_{jj_2}=1$,
there exists $j_1 \in \firstintegers{n}$
such that $A_{jj_1}=1$
and $B_{j_1j_2}=1$.
Then, as $i \refrel_A j$
by supposition,
we get that $A_{i_1j_1}=1$,
which implies
\[
B_{i_2j_2}
=(B \cdot A \cdot B)_{i_2j_2}
= 1,
\textup{ since }
B_{i_2i_1} = A_{i_1j_1}
=
B_{j_1j_2}
=
1.
\]
Since this holds for every
$i_2,j_2 \in \firstintegers{n}$ 
satisfying
$B_{i_2i} = 1 = B_{jj_2}$,
we obtain that
$i \refrel_B j$.
\end{proof}

\begin{lemma}\label{prop::upper_bound}
Let $A$ and $B$
be two idempotent matrices of $\relation{n}$
satisfying $A \cdot B= B = B \cdot A$.
If $A$ and $B$ have the same number
of positive sets and free pairs,
then they are equal.
\end{lemma}

\begin{proof}
Let us pick two idempotent elements
$A,B \in \relation{n}$
such that $A \cdot B= B = B \cdot A$.
Suppose that $A$ and $B$
have the same number of positive sets.
By Lemma~\ref{lemma:clique_containment},
each positive set of $B$ contains
at least one positive set of $A$.
Since the positive sets of $B$
are disjoint, the pigeonhole principle
yields the two following claims.
\begin{customclaim}{1}
Each positive set of $A$ is contained
in a positive set of $B$.
\end{customclaim}
\begin{customclaim}{2}
Each positive set of $B$ contains
exactly one positive set of $A$.
\end{customclaim}
\noindent
Moreover, suppose that $A$ and $B$
have the same number of free pairs.
By Lemma \ref{lemma:incomparable_containment}
every free pair of $B$
is a free pair of $A$.
This yields the following claim.
\begin{customclaim}{3}
The free pairs of $A$ and $B$ are identical.
\end{customclaim}

We now prove that $A = B$.
First, we show that
for every component $A_{ik}$
equal to $1$,
the corresponding component $B_{ik}$
is also equal to $1$.
Since $A$ is stable,
there exists $j \in \firstintegers{n}$
satisfying $A_{ij} = A_{jj} = A_{jk} = 1$.
Then $j$ is contained in a positive set of $A$,
which is itself contained in a positive set of $B$
by Claim~1.
Therefore we obtain that $B_{jj}=1$,
which yields
\[
B_{ik} = (A \cdot B \cdot A)_{ik} = 1,
\textup{ since }
A_{ij}=B_{jj}=A_{jk}=1.
\]

To conclude,
we show that for every component $B_{ij}$
equal to $1$,
the corresponding component $A_{ij}$
is also equal to $1$.
To this end,
we introduce four new elements
$i_1,i_2,j_1,j_2$ in $\firstintegers{n}$:
First, as $(A \cdot B \cdot A)_{ij} = B_{ij} = 1$,
there exist $i_2,j_2 \in \firstintegers{n}$
such that $A_{ii_2}=B_{i_2j_2}=A_{j_2j}=1$.
Second, as $A$ is stable,
there exist $i_1,j_1 \in \firstintegers{n}$
such that
$A_{ii_1}=A_{i_1i_1}=A_{i_1i_2}=1$
and $A_{j_2j_1}=A_{j_1j_1}=A_{j_1j}=1$.
These definitions ensure that
\[
B_{i_1j_1}
=(A\cdot B \cdot A)_{i_1j_1} = 1,
\textup{ since }
A_{i_1i_2}=B_{i_2j_2}=A_{j_2j_1}=1.
\]
Note that, as observed
after the definition
of the relation induced by an idempotent matrix,
this implies that $i_1 \refrel_B j_1$.
We derive from this that either
$i_1 \refrel_A j_1$ or $j_1 \refrel_A i_1$:
if $i_1=j_1$ this follows from the fact that
$\refrel_A$ is reflexive,
and if $i_1 \neq j_1$ 
this follows from Claim~3.
We show that both possibilities
lead to $A_{ij}=1$.
\begin{itemize}
\item 
If $i_1 \refrel_A j_1$,
then we obtain $A_{i_1j_1}=1$
as $A_{i_1i_1}=1=A_{j_1j_1}$.
Therefore,
\[
A_{ij}=(A \cdot A \cdot A)_{ij}=1
\textup{ since }
A_{ii_1} = A_{i_1j_1} = A_{j_1j} = 1.
\]
\item
If $j_1 \refrel_A i_1$,
then we obtain $A_{j_1i_1}=1$
as $A_{j_1j_1}=1=A_{i_1i_1}$.
Therefore,
\[
B_{j_1i_1}=(A \cdot B \cdot A)_{j_1i_1}=1
\textup{ since }
A_{j_1i_1} = B_{i_1j_1} = A_{j_1i_1} = 1.
\]
As a consequence,
$i_1$ and $j_1$ are in
the same positive set of $B$.
Moreover, as $A_{i_1i_1}=A_{j_1j_1}=1$,
both $i_1$ and $j_1$ are elements
of positive sets of $A$.
Combining these two statements with Claim~2
yields that $i_1$ and $j_1$ are
in the same positive set of $A$.
Therefore $A_{i_1j_1}=1$,
which implies that $i_1 \refrel_A j_1$,
and we can conclude as in the previous point.
\end{itemize}
Since we successfully showed that every
$1$ of $A$ corresponds to a $1$ of $B$,
and reciprocally,
we obtain that $A = B$, which proves the statement.
\end{proof}

\subsection{
Lower bound}\label{subsec::low_relation}
We construct a
monomorphism
$\phi$
between the max monoid $\maxSemigroup{f(n)}$,
where $f(n) = \frac{n^2+n+2}{2}$,
and the monoid of Boolean matrices $\relation{n}$.
The construction is split in two steps.
First, we define $\phi$
over the domain
$[1,g(n) + 1]$,
where $g(n) = \frac{n(n-1)}{2}$
is the number of
pairs of elements $i<j$ in $\firstintegers{n}$.
Then,
we complete the definition
over the domain
$[g(n) + 1,f(n)]$.


\subparagraph{Diagonal to triangular}
Let us define $\phi$ over
$[1,g(n)+1]$.
We map the neutral element
$1 \in \maxSemigroup{f(n)}$
to the neutral element
$D_n \in \relation{n}$:
the identity matrix.
Then, we map
$g(n)+1 \in \maxSemigroup{f(n)}$ 
to the full upper triangular matrix
$U_n \in \relation{n}$.
Note that $U_n$ contains $g(n)$
more $1$'s than $D_n$ does.
We define the images
of the elements between $1$ and $g(n)+1$
by gradually adding to $D_n$
the $1$'s of $U_n$ it lacks.
Formally, we order the indices
corresponding to the components above the diagonal
$p_1 < p_2 < \ldots < p_{g(n)}
\in \firstintegers{n} \times \firstintegers{n}$
according to the lexicographic order:
$(i,j)$ comes before $(i',j')$
if either $i<i'$, or $i = i'$ and $j < j'$.
Then, for every $m \in [1,g(n)+1]$,
we construct the image $\phi(m) \in \relation{n}$
as follows:
\begin{itemize}[nolistsep]
    \item 
    Every component $(\phi(m))_{ii}$ of the diagonal
    is $1$;
    \item 
    Every component $(\phi(m))_{ij}$ below the diagonal
    is $0$;
    \item 
    Every component $(\phi(m))_{ij}$ above the diagonal
    is $1$
    if $(i,j) < p_m$, and $0$ otherwise.
\end{itemize}

\subparagraph{Triangular to empty}
Let us define $\phi$ over
$[g(n)+1,f(n)]$.
To fit the first
part of the definition,
we map
$g(n) + 1  \in \maxSemigroup{f(n)}$
to the upper diagonal matrix
$U_n \in \relation{n}$.
Then, we map the absorbing element
$f(n) = g(n) + 1 + n \in \maxSemigroup{f(n)}$
to the absorbing element
$0_n \in \relation{n}$: the null matrix.
Finally, for $m \in [0,n]$,
we construct $\phi(g(n)+1+m)$
by replacing the last $m$ rows of $U_n$
with $0$'s.
Formally, we have:
\begin{itemize}[nolistsep]
    \item 
    Every component $(\phi(g(n)+1+m))_{ij}$
    is $1$
    if $i\leq j$ and $i \leq n-m$,
    and $0$ otherwise.
\end{itemize}

\subparagraph{Proof of correctness}
We prove that the function $\phi$
just defined is a monomorphism.

We show that
$\phi$ is a homomorphism:
$\phi(m) \cdot \phi(m') = \phi(m') =
\phi(m') \cdot \phi(m)$
for
all $1 \leq m \leq m' \leq f(n)$.
First,
note that if
$(\phi(m'))_{ij} = 1$,
then
$(\phi(m) \cdot \phi(m'))_{ij}
= (\phi(m') \cdot \phi(m))_{ij} = 1$:
if $m \leq g(n)+1$,
this follows from the fact that
the diagonal of $\phi(m)$ is filled with $1$'s,
and if $m > g(n)+1$,
since $m \leq m'$ we obtain that
$(\phi(m))_{ii} =(\phi(m'))_{ij} =1
=(\phi(m'))_{ii} =(\phi(m))_{ij}$.
It remains to show that
if
$(\phi(m) \cdot \phi(m'))_{ik} =1$
or
$(\phi(m') \cdot \phi(m))_{ik} = 1$,
then 
$(\phi(m'))_{ik} = 1$.
If $m' \leq g(n)+1$, this holds since
for every triple
$i \leq j \leq k \in \firstintegers{n}$,
the pair $(i,k)$
is lexicographically smaller than
or equal to
$(j,k)$.
If $m' > g(n)+1$,
this holds since
for every triple
$i \leq j \leq k \in \firstintegers{n}$,
trivially $i$ is
smaller than or equal to
both $i$ and $j$.

We conclude by showing that $\phi$ is injective:
between $\phi(1)$ and $\phi(g(n)+1)$ a new
$1$ is added at each step,
and between $\phi(g(n)+1)$ and $\phi(f(n))$
we remove at each step a $1$
of the diagonal
that was present in all the previous
images.

\subparagraph{Example}
We depict the monomorphism
$\phi: \maxSemigroup{f(n)} \rightarrow \relation{n}$
in the case $n=4$
by listing the $f(4)=11$
elements of its image in $\relation{4}$.
Under each element, we state its number of positive sets
followed by its number of free pairs.

  \begin{center}
      
      
    
    \begin{tikzpicture}
      \foreach \c in {0,1,2,3,4,5,6,7,8,9,10}{
        \pgfmathtruncatemacro{\stex}{\c}
        \coordinate (i\c) at
            (1.25*\stex+0.375,-1.1);
        \foreach \d in {0,1,2,3}{
            \foreach \e in {0,1,2,3}{
                \coordinate (a\c\d\e) at
                (1.25*\c + 0.25*\d, -0.25*\e);
                \ifthenelse{\d<\e}
                {\node
                [draw,thick,circle,inner sep=1]
                (m\c\d\e) at (a\c\d\e)
                {};}
                {}
                
            }
            \pgfmathtruncatemacro{\bound}{10-\c}
            \ifthenelse{\d<\bound}
            {\draw[thick] ($(a\c\d\d)+(0,0.06)$)--($(a\c\d\d)-(0,0.06)$);}
            {\node[draw,thick,circle,inner sep=1]
            (m\c\d\d) at (a\c\d\d)
            {};}
        }
        \ifthenelse{0<\c \and \c<10}
        {\draw[thick] ($(a\c10)+(0,0.06)$)--($(a\c10)-(0,0.06)$);}
        {\node[draw,thick,circle,inner sep=1]
        (m\c10) at (a\c10)
        {};}
        \ifthenelse{1<\c \and \c<10}
        {\draw[thick] ($(a\c20)+(0,0.06)$)--($(a\c20)-(0,0.06)$);}
        {\node[draw,thick,circle,inner sep=1]
        (m\c20) at (a\c20)
        {};}
        \ifthenelse{2<\c \and \c<10}
        {\draw[thick] ($(a\c30)+(0,0.06)$)--($(a\c30)-(0,0.06)$);}
        {\node[draw,thick,circle,inner sep=1]
        (m\c30) at (a\c30)
        {};}
        \ifthenelse{3<\c \and \c<9}
        {\draw[thick] ($(a\c21)+(0,0.06)$)--($(a\c21)-(0,0.06)$);}
        {\node[draw,thick,circle,inner sep=1]
        (m\c21) at (a\c21)
        {};}
        \ifthenelse{4<\c \and \c<9}
        {\draw[thick] ($(a\c31)+(0,0.06)$)--($(a\c31)-(0,0.06)$);}
        {\node[draw,thick,circle,inner sep=1]
        (m\c31) at (a\c31)
        {};}
        \ifthenelse{5<\c \and \c<8}
        {\draw[thick] ($(a\c32)+(0,0.06)$)--($(a\c32)-(0,0.06)$);}
        {\node[draw,thick,circle,inner sep=1]
        (m\c32) at (a\c32)
        {};}
      }
     \node[](j0) at (i0) {\scriptsize $(4,6)$};
     \node[](j0) at (i1) {\scriptsize $(4,5)$};
     \node[](j0) at (i2) {\scriptsize $(4,4)$};
     \node[](j0) at (i3) {\scriptsize $(4,3)$};
     \node[](j0) at (i4) {\scriptsize $(4,2)$};
     \node[](j0) at (i5) {\scriptsize $(4,1)$};
     \node[](j0) at (i6) {\scriptsize $(4,0)$};
     \node[](j0) at (i7) {\scriptsize $(3,0)$};
     \node[](j0) at (i8) {\scriptsize $(2,0)$};
     \node[](j0) at (i9) {\scriptsize $(1,0)$};
     \node[](j0) at (i10) {\scriptsize $(0,0)$};
    \end{tikzpicture}
  \end{center}
Starting
with the identity matrix $D_4$,
we gradually add $1$'s,
reaching
the triangular matrix
$U_4$ in $g(4) = 6$ steps.
Then, we erase line after line,
reaching the null matrix $0_4$
in $4$ steps.

\bibliography{biblio}{}

\newpage
\appendix
\section{Known upper bounds for the Ramsey function}\label{subsec:comparison}
We detail the two main methods
used to bound the Ramsey function of a monoid
prior to this work,
and we show some cases in which the bounds obtained are unnecessarily large.

\subparagraph{Ramsey's Theorem}
Given two integers $c$ and $\vari$,
the \emph{multicolour Ramsey number} 
$\textsf{N}_{c}(\vari)$ is
the smallest integer such that
every colouring of a complete graph
on $\textsf{N}_{c}(\vari)$ vertices with $c$ colours
contains a monochromatic clique of size $\vari$.
Ramsey's Theorem~\cite{ramsey2009problem} proves the existence of the
Ramsey number $\textsf{N}_{c}(\vari)$ for every $c, \vari \in \mathbb{N}$.

These numbers can be used to bound
the Ramsey function
associated to a monoid $\calM$:
\begin{equation}\label{eq:RamNum}
\truevalRam{\calM}{\vari} \leq \textsf{N}_{|\calM|}(\vari+1)-1
\textup{ for all } k \geq 2.
\end{equation}
This is proved as follows.
For every word
$u = m_1m_2 \ldots m_n \in \calM^*$
of length $\textsf{N}_{|\calM|}(\vari+1)-1$,
we consider the complete graph 
on $\textsf{N}_{|\calM|}(\vari+1)$ vertices,
and we colour its edges with
the elements of $\calM$
as follows:
for all $1 \leq i < j \leq \textsf{N}_{|\calM|}(\vari)$,
the edge between vertices $i$ and $j$
is coloured with the element
$m_{i} \cdot m_{i+1} \cdot
\ldots \cdot m_{j-1} \in \calM$.
By definition of the Ramsey numbers,
$G$ contains a monochromatic clique of size $\vari+1$,
which corresponds exactly to a Ramsey $\vari$-decomposition of $u$ whenever $\vari \geq 2$.

It is known that $\textsf{N}_{c}(\vari) > 2^{\frac{c\vari}{4}}$
for all $c,\vari \in \mathbb{N}$ (see \cite{lefmann1987note}).
Since $\truevalRam{\calM}{\vari} \leq (\vari|\calM|^{4})^{\reglength{\calM}}$
by Theorem \ref{theorem:general_valram},
the bound \eqref{eq:RamNum} lacks precision 
for large $k$'s,
and also for every monoid $\calM$ with a size $|\calM|$
substantially larger than its regular $\calD$-length
$\reglength{\calM}$.
For instance,
for the full transformation monoid $\transformation{n}$ over $n$ elements 
(see Section \ref{sec:regLength}),
it is exponentially too large:
\[
\begin{array}{ll}
\truevalRam{\transformation{n}}{\vari} \leq (\vari|\transformation{n}|^{4})^{\reglength{\transformation{n}}}
= \vari^{n+1}(n+1)^{4n(n+1)};\\
\textsf{N}_{|\transformation{n}|}(\vari+1) - 1 \geq 2^{\frac{(\vari+1) |\transformation{n}|}{4}}
= 2^{\frac{(\vari+1) (n+1)^n}{4}}.
\end{array}
\]

\subparagraph{Factorisation Forest Theorem}
A \emph{Ramsey factorisation tree} of a word
$u = m_1 \ldots m_n \in \calM^*$
is a directed tree $T$ whose vertices are labelled with
non-empty words of $\calM^+$
such that:
\begin{itemize}[nolistsep]
    \item 
    the root is labelled with $u$;
    \item 
    the leaves are labelled with
    the letters $m_1,m_2,\ldots,m_n$ composing $u$;
    \item 
    each branching is a Ramsey decomposition
    of the parent's label in the children's labels.
\end{itemize}
The Factorisation Forest Theorem states that
for every finite monoid $\calM$
there is a bound $\textsf{F}(\calM)$ such that
each sequence $u \in \calM^*$
has a Ramsey factorisation tree
of height at most $\textsf{F}(\calM)$.
The theorem was initially proved in~\cite{Simon90},
and the bound has been improved afterwards~\cite{ChalopinL04,Kufleitner08,Colcombet2012}.
This result can be used to bound the Ramsey function
associated to $\calM$:
\begin{equation}\label{eq:FacFor}
\truevalRam{\calM}{\vari} \leq (\vari+1)^{\textsf{F}(\calM)}+1
\textup{ for all } k \in \mathbb{N}.
\end{equation}
This is proved as follows.
Given a word $u \in \calM^*$
of length $(\vari+1)^{\textsf{F}(\calM)}+1$,
we consider its Ramsey factorisation tree
of height at most $\textsf{F}(\calM)$.
Since this tree has more than
$(\vari+1)^{\textsf{F}(\calM)}+1$ leaves,
it necessarily contains a branching of size
at least $\vari + 2$.
This yields a Ramsey $\vari$-decomposition of 
a factor of $u$, that can be transferred back to $u$.


It is known that ${\textsf{F}(\calG)} = |\calG|$ for every group $\calG$
(see \cite{Kufleitner08, Colcombet2012},
note that in \cite{Kufleitner08}
this result is proved for $3|\calG|$ instead of $|\calG|$:
the factor $3$ stems from a slightly different definition of $\textsf{F}(\calG)$).
As a consequence, 
since $\truevalRam{\calM}{\vari} \leq (\vari|\calM|^{4})^{\reglength{\calM}}$
by Theorem \ref{theorem:general_valram},
the bound \eqref{eq:FacFor}
lacks precision for every monoid $\calM$
containing a group $\calG$ substantially larger than its regular $\calD$-length
$\reglength{\calM}$.
For instance, as the full permutation group $P_n$ over $n$ elements is embedded into
the full transformation monoid $\transformation{n}$ (see Section \ref{sec:regLength}),
the bound \eqref{eq:FacFor} is exponentially too large:
\[
\begin{array}{ll}
\truevalRam{\transformation{n}}{\vari} \leq (\vari|\transformation{n}|^{4})^{\reglength{\transformation{n}}}
= \vari^{n+1}(n+1)^{4n(n+1)};\\
(\vari+1)^{\textsf{F}(\transformation{n})} +1
\geq (\vari+1)^{\textsf{F}(P_n)} +1
\geq (\vari+1)^{|P_n|} +1
=  (\vari+1)^{n!} +1.
\end{array}
\]

\section{
Defining the regular $\calD$-length
using Green's relations}\label{app:Green}

We state basic definitions and lemmas
concerning Green's relations of finite monoids.
More details can be found
in~\cite{pin2010mathematical}
(note that here we write $\calD$ instead of $\calJ$).

The preorders $\leq_{\calH}$
and $\leq_{\calD}$ over a finite monoid $\calM$
are defined as follows:
\[
\begin{array}{ll}
m \leq_{\calH} m'
\textup{ if }
s \cdot m' = m = m' \cdot t
\textup{ for some }
s,t \in \calM;\\
m \leq_{\calD} m'
\textup{ if }
m = s \cdot m' \cdot t
\textup{ for some }
s,t \in \calM.\\
\end{array}
\]
An $\calH$-class of $\calM$
is an equivalence class of the equivalence
relation $\sim_\calH$
generated by $\leq_\calH$.
Similarly, a $\calD$-class
is an equivalence class of the equivalence
relation $\sim_\calD$
generated by $\leq_\calD$.
A $\calD$-class is called regular if it contains
at least one idempotent element of $\calD$.
We denote by $\calD(m)$ the $\calD$-class of
an element $m \in \calM$.
We use the two following lemmas.

\begin{lemma}\label{lemma_Rclass}
Let $m,m'\in \calM$
be two elements of the same $\calD$-class.
If $s \cdot m' = m = m' \cdot t$
for some $s,t \in \calM$,
then $m \sim_{\calH} m'$.
\end{lemma}

\begin{lemma}\label{lemma_Hclass_idem}
Each $\calH$-class of $\calM$ contains at most
one idempotent element.
\end{lemma}
\noindent
We prove the equivalence of both definitions
of the regular $\calD$-length
stated in the introduction.

\begin{proposition}
The two following definitions are equivalent
\begin{itemize}[nolistsep]
    \item
    $\reglength{\calM}$
    is the size of the largest max monoid
    $\maxSemigroup{\reglength{\calM}}$
    embedded in $\calM$;
    \item
    $\reglength{\calM}$ is the size
    of the largest chain of regular $\calD$-classes
    of $\calM$.
\end{itemize}
\end{proposition}

\begin{proof}
Transforming a monomorphism
$\phi: \maxSemigroup{n} \rightarrow \calM$
into a chain of regular $\calD$-classes is easy.
Remark that, by definition, the order $<_{\calD}$ over $\maxSemigroup{n}$
is the inverse of the usual order:
we have
$n
<_{\calD}
n-1
<_{\calD} \ldots <_{\calD}
1$.
We show that
$\calD(\phi(n))
<_{\calD}
\calD(\phi(n-1))
<_{\calD} \ldots <_{\calD}
\calD(\phi(1))$.
First, for every $1 \leq i < j \leq n$,
$\phi(j) \leq_{\calD} \phi(i)$
since $\phi(i) \cdot \phi(j) = \phi(j)$.
Moreover, $\phi(i)$ and $\phi(j)$ are not in the
same $\calD$-class,
as the fact that
$\phi(i) \phi(j) = \phi(j) = \phi(j) \phi(i)$
would imply that $\phi(i)$ and $\phi(j)$
are in the same
$\calH$-class by Lemma
\ref{lemma_Rclass}.
Then $\phi(i)$ and $\phi(j)$ would be
equal by Lemma 
\ref{lemma_Hclass_idem},
which contradicts the fact that $\phi$ is injective.

Transforming a chain of $n$ regular $\calD$-classes
into a monomorphism
$\phi: \maxSemigroup{n} \rightarrow \calM$
requires a bit of work:
let
${D_n} <_{\calD} {D_{n-1}}
<_{\calD} \ldots
<_{\calD} {D_1}$
be a chain of regular $\calD$-classes.
For every $1 \leq i \leq n$,
as the $\calD$-class ${D_i}$
is regular, it contains
at least one idempotent element $f_i$.
Unfortunately, there is no guarantee
that the elements $f_i$
form a submonoid of $\calM$.
We now show how
to transform each $f_i$
into an idempotent element $e_i$
satisfying $e_i \sim_\calD f_i$
such that the function
$\phi: \maxSemigroup{n} \rightarrow \calM$
mapping $i$ to $e_i$ is a monomorphism.
First, we set $e_1 = f_1$.
Then for every $1<i\leq n$
we construct
$e_{i+1}$ based on $e_i$.
Since $f_{i+1} <_{\calD} f_i \sim_{\calD} e_i$,
there exists $s,t \in \calM$ such that
$f_{i+1} = s \cdot e_i \cdot t$.
We prove that setting
$e_{i+1} =
e_i \cdot t \cdot f_{i+1} \cdot s \cdot e_i$
satisfies the desired properties.
First, $e_{i+1}$ is an idempotent element of $\calM$ since
\begin{equation}\label{equ:JchainHAHA}
e_{i+1} \cdot e_{i+1}
= e_i \cdot t \cdot f_{i+1} \cdot s \cdot e_i \cdot
e_i \cdot t \cdot f_{i+1} \cdot s \cdot e_i
= e_i \cdot t \cdot f_{i+1} \cdot s \cdot e_i
= e_{i+1}.
\end{equation}
Moreover, $e_{i+1} \sim_{\calD} f_{i+1}$
since $e_{i+1} \leq_{\calD} f_{i+1}$
by definition,
and
\begin{equation}\label{equ:JchainHAHAHA}
s \cdot e_{i+1} \cdot t
= s \cdot e_i \cdot t \cdot f_{i+1} \cdot s \cdot e_i \cdot t
= f_{i+1} \cdot f_{i+1} \cdot f_{i+1}
= f_{i+1}.
\end{equation}
Finally,
$e_{i+1} \cdot e_i = e_{i+1} = e_i \cdot e_{i+1}$ as
\begin{equation}\label{equ:Jchain}
e_{i+1} \cdot e_{i}
= e_i \cdot t \cdot f_{i+1} \cdot s \cdot e_i \cdot e_i
= e_{i+1}
= e_i \cdot e_i \cdot t \cdot f_{i+1} \cdot s \cdot e_i
= e_i \cdot e_{i+1}.
\end{equation}
As a consequence, the function
$\phi: \maxSemigroup{n} \rightarrow \calM$ mapping 
each $1 \leq i \leq n$ to $e_i$ is a monomorphism:
it is a homomorphism by
Equations \ref{equ:JchainHAHA} and \ref{equ:Jchain},
and it is injective since
the elements of its image are all in distinct $\calD$-classes by 
Equation \ref{equ:JchainHAHAHA}.
\end{proof}

\section{Properties
of idempotent Boolean matrices}\label{app:mat}
We prove the technical properties
of idempotent Boolean matrices
stated in Section \ref{sec:reglength}.

\begin{lemma}\label{lemma:characterisation_idempotent}
Every idempotent matrix is stable.
\end{lemma}

\begin{proof}
Let $A \in \relation{n}$ be an idempotent matrix,
and let $A_{ik}$ be a component of $A$ equal to $1$.

We begin by defining inductively a sequence 
of $n+1$ elements
$j_0,j_1,\ldots,j_n \in \firstintegers{n}$
such that
\begin{enumerate*}[(a)]
\item
$A_{ij_s} =1$
for all $0 \leq s \leq n$,
\item
$A_{j_sj_{t}}=1$
for all $0 \leq t < s \leq n$.
\end{enumerate*}
First, setting $j_0 = k$ ensures that $A_{ij_0} =1$.
Now, let $0 \leq s < n$,
and suppose that $j_s$ satisfies
the desired properties.
Since $(A \cdot A)_{ij_s} = A_{ij_s}=1$,
there exists 
$j_{s+1} \in \firstintegers{n}$
satisfying 
$A_{ij_{s+1}}=A_{j_{s+1}j_{s}}=1$.
We immediately obtain that
$A_{ij_{s+1}}=1$.
Moreover, for every $t < s+1$,
either $t = s$ and
$A_{j_{s+1}j_{s}}=1$,
or $t<s$
and $A_{j_{s+1}j_{t}}=(A \cdot A)_{j_{s+1}j_{t}}=1$
since $A_{j_{s+1}j_{s}} = A_{j_sj_{t}} = 1$.

As all the $j_s$ are in $\firstintegers{n}$,
there exist two indices $0 \leq s < t \leq n$
satisfying $j_s = j_t$.
Then, setting $j=j_s=j_t$ yields
$A_{ij}=A_{ij_t}=1$,
$A_{jj} = A_{j_tj_s}=1$,
and
$A_{jk} = A_{j_sj_0} =1$.
This proves that the matrix $A$ is stable.
\end{proof}

\begin{lemma}
The positive sets of an idempotent Boolean matrix
are disjoint.
\end{lemma}

\begin{proof}
Let $A \in \relation{n}$ be a Boolean matrix,
and let $I$ and $J$ be two positive sets of $A$.
To prove the statement, we show that if the intersection
of $I$ and $J$ is not empty, then they are equal.

Suppose that there exists $k \in I \cap J$.
Then for every $i \in I$, for every $j \in J$,
\[
\begin{array}{l}
A_{ij} = (A \cdot A)_{ij}
= 1 \textup{ since }
A_{ik} = A_{kj} = 1;\\
A_{ji} = (A \cdot A)_{ji}
= 1 \textup{ since }
A_{jk} = A_{ki} = 1.
\end{array}
\]
Since positive sets are maximal by definition,
this proves that $I = J$.
\end{proof}

\begin{lemma}
For every idempotent matrix
$A\in \relation{n}$,
the relation $\refrel_A$ is reflexive.
\end{lemma}

\begin{proof}
Let $A \in \relation{n}$ be an idempotent matrix.
For every $j \in \firstintegers{n}$,
for every $i,k \in \firstintegers{n}$
satisfying $A_{ij} = 1 = A_{jk}$,
\[
A_{ik} = (A \cdot A)_{ik}
= 1 \textup{ since }
A_{ij} = A_{jk} = 1.
\]
Therefore $j \refrel_A j$, which proves
that $\refrel_A$ is reflexive.
\end{proof}

\begin{lemma}
For every idempotent matrix
$A\in \relation{n}$,
$A_{ij} = 1$ implies $i \refrel_A j$.
\end{lemma}

\begin{proof}
Let $A \in \relation{n}$ be an idempotent matrix,
and let us pick a component
$A_{ij}$ equal to $1$.
Then for every $i_2,j_2 \in \firstintegers{n}$
satisfying $A_{i_2i} = 1 = A_{jj_2}$,
\[
A_{i_2j_2} = (A \cdot A \cdot A)_{i_2j_2}
= 1 \textup{ since }
A_{i_2i} = A_{ij} = A_{jj_2} = 1.
\]
Therefore $i \refrel_A j$, which concludes the proof.
\end{proof}

\end{document}